\documentclass{amsart}
\usepackage{amsmath}

\begin{document}
\newtheorem{theorem}{Theorem}
\newtheorem{lemma}{Lemma}
\newtheorem{corollary}{Corollary}
\newcommand{\err}{e}
\newcommand{\li}{{\rm li}}

\title{Approximately Counting Semismooth Integers}


\author{Eric Bach}
 \address{Computer Sciences Department,
 University of Wisconsin-Madison}
 \email{bach@cs.wisc.edu}
\author{Jonathan Sorenson}
  \address{Computer Science and Software Engineering Department,
  Butler University}
  \email{sorenson@butler.edu}

\date{24 April 2013}

\maketitle

\begin{abstract}
An integer $n$ is $(y,z)$-semismooth if $n=pm$ where
  $m$ is an integer with all prime divisors $\le y$ and 
  $p$ is 1 or a prime $\le z$.
Large quantities of semismooth integers are utilized
  in modern integer factoring algorithms,
  such as the number field sieve,
  that incorporate the so-called \textit{large prime} variant.
Thus, it is useful for factoring practitioners to be able to estimate
  the value of $\Psi(x,y,z)$, the number of $(y,z)$-semismooth integers
  up to $x$, so that they can better set algorithm parameters and
  minimize running times, 
  which could be weeks or months on a cluster supercomputer.
In this paper, we explore several algorithms to approximate 
  $\Psi(x,y,z)$ using a generalization of Buchstab's identity
  with numeric integration.
\end{abstract}

%
%

\section{Introduction}

The security of the public-key cryptosystem RSA \cite{MOV,RSA78}
  is based on the practical difficulty of integer factoring.

The fastest current general-purpose integer factoring algorithm 
  is the number field sieve \cite{CP,Lenstra}, which in its basic form
  makes use of \textit{smooth numbers}, integers with only small
  prime divisors.
This has inspired research into algorithms to approximately count
  smooth numbers 
\cite{Bernstein98,Bernstein2002,HS97,PS06,Sorenson2000,Suzuki2004,Suzuki2006}.
However, most implementations of the number field sieve make use of the
  so-called \textit{large prime} variant \cite[\S6.1.4]{CP}.
So we want, in fact, to count smooth integers that admit at most one slightly
  larger prime divisor, or \textit{semismooth} numbers.
(See, for example, the details on the factorization of
  a 768-bit RSA modulus \cite{factor768bitRSA} where smoothness
  bounds are discussed near the end of \S2.2.)
\nocite{Granville2008,Moree2013}

The principal contribution of this paper is twofold:
\begin{enumerate}
  \item
    We present data showing that the key to estimating
    $\Psi(x,y,z)$ accurately is an algorithm to estimate $\Psi(x,y)$
    accurately, and
  \item
    We present head-to-head comparisons of five algorithms
    for estimating $\Psi(x,y,z)$.
\end{enumerate}

Previous work was done by Bach and Peralta \cite{BP96} and generalized by
  Zhang and Ekkelkamp \cite{Ekkelkamp2007,zhang}; we discuss this below.

This paper is organized as follows.
We begin with some definitions, and briefly discuss computing exact counts
  of semismooth integers.
We then give our main theoretical result,
  a generalized Buchstab identity, which
  together with numerical integration, is the basis of all our algorithms.
We then present five different algorithms in some detail,
  two based on the Dickman $\rho$ function and three based on
  the saddle point methods of Hildebrand and Tenenbaum \cite{HT93},
  along with empirical results for each algorithm.
As one might expect, we discover a tradeoff in algorithm choice
  between speed and accuracy.
We follow this up with an elaboration on some numerical details.

\begin{table*} 
\caption{Exact Values of $\Psi(x,y,z)$, $x=2^{40}$\label{tableexact}}
{

{
\begin{center}\begin{tabular}{l|rrrrrrr}
\hline
$y$ &
$z=2^{10}$ & $z=2^{12}$ & $z=2^{14}$ & $z=2^{16}$ & $z=2^{18}$ & $z=2^{20}$ & \\
\hline
$2^{2}$  &
58916 & 170906 & 503392 & 1500366 & 4513650 & 13597105 \\
$2^{4}$  &
 6132454 & 15111450 & 36766896 & 88920834 & 213965871 & 508848834  \\
$2^{6}$  &
 323105012 & 678707129 & 1326493628 & 2499496319 & 4603776946 & 8298713253  \\
$2^{8}$  &
 3157707079 & 6694272918 & 11837179134 & 19296840890 & 30059136386 & 45290571262  \\
$2^{10}$  &
 7138986245 & 21494669620 & 39400743040 & 61719198990 & 89501569374 & 123782024151  \\
$2^{12}$  &
 --- & 30641713551 & 68600140477 & 111769092210 & 160884758713 & 215725604647 \\
$2^{14}$  &
 --- & --- & 80324574755 & 145583683889 & 214469637137 & 286977146180 \\
$2^{16}$  &
 --- & --- & --- & 155283653287 & 241316058768 & 329068435579 \\
$2^{18}$  &
 --- & --- & --- & --- & 248857736183 & 349745847766 \\
$2^{20}$  &
 --- & --- & --- & --- & --- & 354983289990 \\
\hline
\end{tabular}
\end{center}
}
}
\end{table*}

\section{Definitions}

Let $P(n)$ denote the largest prime divisor of the positive integer $n$,
  with $P(1)=1$.
An integer $n$ is \textit{$y$-smooth} if $P(n)\le y$, and $\Psi(x,y)$
  counts the integers $n\le x$ that are $y$-smooth.

An integer $n$ is \textit{$(y,z)$-semismooth} if we can write
  $n=mp$ where $m$ is $y$-smooth and $p\le z$ is a prime or 1.
$\Psi(x,y,z)$ counts the integers $n\le x$ that are $(y,z)$-semismooth.
(Generalizations to more than one exceptional prime have been
defined by Zhang and Ekkelkamp \cite{Ekkelkamp2007,zhang}.)
Observe that $\Psi(x,y,y)=\Psi(x,y)$, the function
  $\Psi(x,y,x)$ counts integers whose second-largest prime divisor is
  bounded by $y$, and $\Psi(x,1,z) = \min\{ \pi(x), \pi(z) \}$, where
  $\pi(x)$ is the number of primes up to $x$.

Our basic unit of work is the floating point operation.  
Along with the four basic arithmetic operations ($+,-,\times,\div$) 
  we include square roots, logarithms, and exponentials, since their
  complexity is close to that of multiplication
  (see for example \cite{Brent76}).

\section{Exact Counts}

Using a prime number sieve, such as the sieve of Eratosthenes, 
  we can completely factor all integers up to $x$ in
  $O(x\log\log x)$ arithmetic operations, 
  and thereby compute exact values of $\Psi(x,y,z)$.
Of course this is not a practical approach for large $x$,
  but it is useful for evaluating the accuracy of approximation algorithms,
  which is what we do here.
So we wrote a program to do this, based on a segmented sieve of
  Eratosthenes (see \cite{Sorenson06} for prime number sieve references),
  and we ran our program up to
  $x=1099511627776=2^{40}$.
Our results for this largest value for $x$ appear in Table \ref{tableexact},
  which took just over 100 CPU hours to compute.

\section{A Generalized Buchstab Identity}

We have the following version of Buchstab's identity 
  (see for example \cite[p.~365]{Tenenbaum}):
\begin{equation}
  \Psi(x,y) = \Psi(x,2) + \sum_{2<p\le y} \Psi(x/p,p),
\end{equation}
which is obtained by summing over the largest prime divisor of
$y$-smooth integers $n\le x$.
Using this same idea gives us the following:
\begin{equation}\label{gbi}
\Psi(x,y,z) = \Psi(x,y) + \sum_{y<p\le z} \Psi(x/p,y).
\end{equation}
As one can see, the identity is obtained by summing over the
  largest prime divisor.

\begin{table*}
  \caption{ \label{tablesigma}
  $x \cdot \sigma(u,v) / \Psi(x,y,z)$, $x=2^{40}$ 
  }
  {
\begin{center}\begin{tabular}{l|lllllll}
\hline
$y$ &
$z=2^{10}$ & $z=2^{12}$ & $z=2^{14}$ & $z=2^{16}$ & $z=2^{18}$ & $z=2^{20}$  \\
\hline
$2^{2}$ &
7.0984e-14 & 1.2946e-12 & 2.1963e-11 & 3.4242e-10 & 4.8536e-09 & 6.2227e-08 & \\
$2^{4}$ &
0.0042182 & 0.007326 & 0.012701 & 0.021657 & 0.035951 & 0.058 & \\
$2^{6}$ &
0.26627 & 0.2956 & 0.32972 & 0.36872 & 0.41167 & 0.4585 & \\
$2^{8}$ &
0.64392 & 0.66898 & 0.68836 & 0.70731 & 0.72686 & 0.74754 & \\
$2^{10}$ &
0.75636 & 0.80963 & 0.82644 & 0.83778 & 0.84679 & 0.85628 & \\
$2^{12}$ &
 --- & 0.84863 & 0.87886 & 0.89096 & 0.89979 & 0.90729 & \\
$2^{14}$ &
 --- &  --- & 0.89495 & 0.9169 & 0.92614 & 0.93196 & \\
$2^{16}$ &
 --- &  --- &  --- & 0.92275 & 0.93539 & 0.9421 & \\
$2^{18}$ &
 --- &  --- &  --- &  --- & 0.93769 & 0.94722 & \\
$2^{20}$ &
 --- &  --- &  --- &  --- &  --- & 0.95043 & \\
\hline
\end{tabular}\end{center}
  }
\end{table*}

\begin{table*}
  \caption{ \label{tableekkel}
  $E(x,y,z) / \Psi(x,y,z)$, $x=2^{40}$ 
  }
  {
\begin{center}\begin{tabular}{l|lllllll}
\hline
$y$ &
$z=2^{10}$ & $z=2^{12}$ & $z=2^{14}$ & $z=2^{16}$ & $z=2^{18}$ & $z=2^{20}$ & \\
\hline
$2^{2}$ &
1.6195e-13 & 2.9202e-12 & 4.8929e-11 & 7.5256e-10 & 1.0508e-08 & 1.325e-07 & \\
$2^{4}$ &
0.0063754 & 0.010979 & 0.018858 & 0.031834 & 0.052265 & 0.0833 & \\
$2^{6}$ &
0.34328 & 0.37895 & 0.42001 & 0.46628 & 0.51656 & 0.57018 & \\
$2^{8}$ &
0.76655 & 0.79204 & 0.81087 & 0.82903 & 0.84748 & 0.86587 & \\
$2^{10}$ &
0.87052 & 0.91765 & 0.93094 & 0.93899 & 0.94491 & 0.95185 & \\
$2^{12}$ &
 --- & 0.94459 & 0.96506 & 0.97334 & 0.97836 & 0.98148 & \\
$2^{14}$ &
 --- &  --- & 0.97446 & 0.98537 & 0.98863 & 0.99038 & \\
$2^{16}$ &
 --- &  --- &  --- & 0.98694 & 0.99092 & 0.99305 & \\
$2^{18}$ &
 --- &  --- &  --- &  --- & 0.99154 & 0.99516 & \\
$2^{20}$ &
 --- &  --- &  --- &  --- &  --- & 0.99766 & \\

\hline
\end{tabular}\end{center}
  }
\end{table*}

\section{Approximate Counts}

As mentioned in the Introduction, there are many algorithms
  to estimate values of $\Psi(x,y)$.
We could choose one of them, compute a list of primes up to $z$,
  and then apply (\ref{gbi}) to approximate $\Psi(x,y,z)$.
We found that this does, in fact, give fairly accurate estimates, 
  but the resulting
  algorithms are quite slow since roughly $O(z/\log z)$
  evaluations of $\Psi(x/p,y)$ (one for each $p$, $y<p\le z$) are needed.

Our approach, then, is to replace the sum in (\ref{gbi}) 
  with an integral, and then use numeric integration 
  to evaluate it \cite[\S7.2]{CdB}; we used Simpson's rule.
We found that in practice, the relative error introduced by replacing the
  sum with an integral that was then estimated, was less than the
  relative error introduced by the approximation algorithms for $\Psi(x,y)$.

Let us define $\li(x):=\int_2^x dt/\log t$,
and let $\err(x):=\pi(x)-\li(x)$.  By the prime number theorem,
$\err(x)=x/\exp[\Omega(\sqrt{\log x})]$;
if we assume the Riemann Hypothesis, $\err(x)=O(\sqrt{x}\log x)$
\cite{Schoenfeld76}.

We have the following (see \cite[\S2.7]{BS}):
\begin{lemma} \label{lemmaBS}
  Let $f$ be a continuously differentiable function on an open
  interval containing $[2,z]$, and let $2\le y\le z$.
  Then 
  \begin{eqnarray*}
    \sum_{y<p\le z} f(p) &=& \int_y^z \frac{f(t)}{\log t}dt \\
      && +f(z)\err(z)-f(y)\err(y) -\int_y^z \err(t)f^\prime(t)dt.
  \end{eqnarray*}
\end{lemma}
Of course we cannot apply this lemma to $\Psi(x,y)$ directly,
  so we use an estimate instead.
Define $\rho(u)$ as the unique continuous solution to
\begin{eqnarray*}
\rho(u) &=& 1 \qquad (0\le u\le 1) \\
\rho(u-1)+u\rho^\prime(u) &=& 0 \qquad (u>1).
\end{eqnarray*}
Note that $\rho \in C^1$ for $u > 1$.
Hildebrand \cite{Hildebrand86} proved that
for $\epsilon>0$ we have
\begin{equation} \label{xrho}
  \Psi(x,y) = x \rho(u)
   \left(1+O_\epsilon\left(\frac{\log (u+1)}{\log y}\right)\right),
\end{equation}
uniformly on the set defined by 
$1\le u\le \exp((\log y)^{3/5-\epsilon})$
and $y\ge2$.  Here, $u:=u(x,y)=\log x/\log y$.

\begin{theorem} \label{thmintegrate}
Given $2\le y\le z\le x$, let $\epsilon>0$, and assume that
$1\le \log (x/z)/\log y \le \log x/\log y \le \exp((\log y)^{3/5-\epsilon})$.
Then
\begin{equation}
   \Psi(x,y,z)=\left( \Psi(x,y)+\int_y^z \frac{\Psi(x/t,y)}{\log t} dt\right)
     (1+o(1)).
\end{equation}
For asymptotic notation, we are assuming $y$ is large.
\end{theorem}
\begin{proof}
Define $f(t):= (x/t) \rho( \log(x/t)/\log y )$.
By (\ref{xrho}) we have $f(p)=\Psi(x/p,y)(1+o(1))$
  for all primes $y<p\le z$.
$f$ is differentiable and continuous, with 
  $f^\prime(t) \sim -f(t)/t$ (for large $y$).
It is then straightforward to show that
  $\err(t)f^\prime(t) = o( f(t)/\log t)$.
Also since $f$ is decreasing we can show
  $\err(z)f(z)-\err(y)f(y) 
    =o(\pi(z)-\pi(y)) f(z) + O(|\err(y)|f(y))$
  and then
    $(\pi(z)-\pi(y)) f(z)\le\sum_{y<p\le z} f(p) $
  and $|\err(y)|f(y)=o(\Psi(x,y))$.
In the case when $u$ is large, we make use of
  Lemma 8.1 and (61) from \cite[\S5.4]{Tenenbaum}.
%
We then apply Lemma \ref{lemmaBS} and use (\ref{gbi}),
  and finally substitute $\Psi(x/t,y)$ back in for $f(t)$ to complete
  the proof.
\end{proof}

It would be nice to have some function $g(y,z)$ where
\begin{equation}
  \Psi(x,y,z) \approx \Psi(x,y) \cdot g(y,z)
    \label{futurework}
\end{equation}
if this is possible.
Rewriting (\ref{gbi}) we have
\begin{equation}
  \Psi(x,y,z)=\Psi(x,y) \cdot
     \left(1+\sum_{y<p\le z} \frac{\Psi(x/p,y)}{\Psi(x,y)}\right).
\end{equation}
A very crude estimate of $\Psi(x/p,y)/\Psi(x,y)\approx 1/p$
  leads to
\begin{equation}
       \Psi(x,y,z)\approx \Psi(x,y) \cdot
         \left(1+\log( \log z/\log y)) \right).
\end{equation}
In practice, this is too crude to be useful.
However, this estimate can certainly be improved using, for example,
  Theorem 11 from \cite[\S5.5]{Tenenbaum}.
This is a possible direction for future work.

\begin{table*}
  \caption{$HT(x,y,z) / \Psi(x,y,z)$, $x=2^{40} $ \label{tableht} }
  {
\begin{center}\begin{tabular}{l|lllllll}
\hline
$y$ &
$z=2^{10}$ & $z=2^{12}$ & $z=2^{14}$ & $z=2^{16}$ & $z=2^{18}$ & $z=2^{20}$ & \\
\hline
$2^{2}$ &
1.0969 & 1.0712 & 1.0641 & 1.067 & 1.0728 & 1.0833 & \\
$2^{4}$ &
1.0592 & 1.039 & 1.031 & 1.0353 & 1.0507 & 1.0829 & \\
$2^{6}$ &
1.0414 & 1.0281 & 1.0205 & 1.024 & 1.0566 & 1.1582 & \\
$2^{8}$ &
1.0194 & 1.0165 & 1.0135 & 1.0127 & 1.0296 & 1.1338 & \\
$2^{10}$ &
1.0024 & 1.0104 & 1.0101 & 1.01 & 1.0112 & 1.0437 & \\
$2^{12}$ &
 --- & 1.0043 & 1.0104 & 1.01 & 1.0078 & 1.0107 & \\
$2^{14}$ &
 --- &  --- & 1.0046 & 1.0084 & 1.0115 & 1.0143 & \\
$2^{16}$ &
 --- &  --- &  --- & 1.0101 & 1.016 & 1.017 & \\
$2^{18}$ &
 --- &  --- &  --- &  --- & 1.0142 & 1.0116 & \\
$2^{20}$ &
 --- &  --- &  --- &  --- &  --- & 1.0083 & \\
\hline
\end{tabular}\end{center}
  }
\end{table*}

\subsection{The Method of Bach and Peralta}

The first algorithm to try would be to use the estimate
  $\Psi(x,y)\approx x\rho(u)$ from (\ref{xrho})
  and plug it into Theorem \ref{thmintegrate}.
This, in fact, is simply another way to derive the algorithm
  of Bach and Peralta \cite{BP96}.
They define
  $u:=\log x/\log y$, $v:=\log x/\log z$ and
\begin{equation} \label{sigma}
  \sigma(u,v) :=
           \rho(u) + \int_v^u (\rho(u-u/w)/w) dw.
\end{equation}
They then prove that for fixed $u,v$ and $x \rightarrow \infty$,
\begin{equation}
  \Psi(x,y,z) \approx x \sigma(u,v).
\end{equation}
We can use (\ref{gbi}) to obtain the same approximation, as follows:
  \begin{eqnarray*}
       \Psi(x,y,z)&\approx& \Psi(x,y)+\int_y^z (\Psi(x/t,y)/\log t) dt \\
           &\approx& x \cdot \rho(\log x/\log y)  \\
           &&       + \int_y^z (x/t\log t) \rho( \log(x/t)/\log y) dt \\
           &=& x \cdot \left( \rho(u) + \int_v^u (\rho(u-u/w)/w) dw \right) \\
           &=& x \cdot \sigma(u,v).
\end{eqnarray*}
We adapted code written by Peralta to compute values of the
  Dickman $\rho$ function (accurate to roughly 8 decimal digits)
  to compute $\sigma$.
Bach and Peralta discuss methods to compute both $\rho$ and $\sigma$
  in some detail in \cite{BP96}.
We present the results from this algorithm in Table \ref{tablesigma}.
As for all our algorithms, we show the ratio of what the algorithm
  produces as an estimate divided by the actual values we computed earlier.
The closer we are to 1 the better the estimates.

This algorithm is fast, but the results are not as good as we might desire.
As one might expect, they are about as good as what was found
  for estimating counts of smooth numbers using (\ref{xrho})
  in \cite{HS97}.

\subsection{Ekkelkamp's Improvement}

Ekkelkamp \cite{Ekkelkamp2007}  pointed out that $\sigma(u,v)$ could
be made more accurate by adding a quantity which, in our
notation, is
$$
\frac {(1 - \gamma) x}{\log x}
\left[
\rho(u-1) + \int_v^u \frac{\rho(u - u/w - 1)}{w-1} dw
\right].
$$
This can be derived by using the better approximation
$$
\Psi(x,y) \approx x\rho(u) + \frac {(1 - \gamma) x}{\log x}\rho(u-1),
$$
due to Ramaswami \cite[Theorem 1]{Ramaswami49}.
To get her formula, substitute $w = 1/\lambda$ in the integral; 
note that we need
the first term inside the brackets, since our definition of
semismoothness differs from hers.  
(Our ``large prime'' can be 1.)  
This correction does not require much additional effort;
essentially just one more numerical integration.
Let $E(x,y,z)$ denote this approximation.

Our experiments with this indicate that the additional term enhances
accuracy significantly when $y$ and $z$ are large.  
This is roughly the bottom corner of Table \ref{tablesigma}.  
However, the accuracy still drops off
dramatically for smaller values of $y$.  
See Table \ref{tableekkel}.

\subsection{A Saddlepoint Method}

Our third algorithm is based on Algorithm HT
for estimating $\Psi(x,y)$ presented in \cite{HS97}.
Define
\newcommand{\ubar}{\bar{u}}
\begin{eqnarray*}
  \zeta(s,y)&:=&
    \prod_{p\le y} (1-p^{-s})^{-1}; \\
  \phi(s,y)&:=&
    \log \zeta(s,y); \\
  \phi_k(s,y)&:=&
    \frac{d^k}{ds^k} \phi(s,y) \qquad (k\ge 1).
\end{eqnarray*}
The functions $\phi_k$ can be expressed as sums over primes.
Indeed, we have
\begin{eqnarray*}
\phi  (s,y) &=& -\sum_{p\le y} \log (1 - p^{-s}); \\
\phi_1(s,y) &=& -\sum_{p\le y} \frac{\log p}{p^s-1}; \\
\phi_2(s,y) &=& \sum_{p\le y}
  \frac{p^s(\log p)^2}{(p^s-1)^2}.  \\
\end{eqnarray*}
Thus, with a list of primes up to $y$, the quantities 
$\zeta(s,y)$, $\phi_1(s,y)$, and
$\phi_2(s,y)$ can be computed in $O(y/\log y)$ floating point operations.

Define
$$
HT(x,y,s) := 
\frac{ x^s \zeta(s,y) }{ s \sqrt{ 2\pi\phi_2(s,y) }},
$$
and let $\alpha$ be the unique solution to
$ \phi_1(\alpha,y)+\log x =0$.
Hildebrand and Tenenbaum proved the following \cite{HT86}:
\begin{theorem} \label{thmht}
\begin{equation}
  \Psi(x,y) = HT(x,y,\alpha) \cdot 
  \left(1+O\left(\frac{1}{u} + \frac{(\log y)}{y}\right)\right)
\end{equation}
uniformly for $2\le y\le x$.
\end{theorem}

This gives us Algorithm HT \cite{HS97}:
\begin{enumerate}
  \item Find the primes up to $y$.
  \item Compute an approximation $\alpha^\prime$ to $\alpha$
        using binary search and Newton's method.
        Make sure that $|\alpha-\alpha^\prime|=O(1/(u\log x))$.
  \item Output $HT(x,y,\alpha^\prime)$.
\end{enumerate}
Write $HT(x,y)$ for the value output in the last step.
The running time is
$O\left( \frac{y\log\log x}{\log y} + \frac{y}{\log\log y}\right)$
floating point operations.

We simply plugged Algorithm HT into Theorem \ref{thmintegrate}
  to estimate $\Psi(x,y,z)$ using the saddle point method
  as follows:
\begin{equation}
  HT(x,y,z) := HT(x,y)+ \int_y^z (HT(x/t,y)/\log t) dt.
\end{equation}
In Table \ref{tableht} we give the results for this algorithm,
  which are quite good.
The method is, however, a bit slow.

Using the summation algorithms described in \cite{Bach2007},
  we can lower the exponent of $y$ in the running time from 1 to 2/3.  
We give some details for this in \S\ref{secATM}.
We did not implement this improvement, however, 
  because it would not change the computed results, 
  and the method of the section below is faster.

\subsection{Assuming Riemann's Hypothesis\label{htfast}}

This is the same as Algorithm HT, 
  only sums over primes ($\zeta,\phi_1,\phi_2$) above roughly $\sqrt{y}$ are 
  estimated using the prime number theorem plus
  the Riemann Hypothesis \cite{Sorenson2000,Schoenfeld76}.
It is \textit{much} faster than Algorithm HT and nearly as accurate;
  its running time is roughly $\sqrt{y}$.
Let $HT_f(x,y)$ denote the estimate this algorithm computes for
  $\Psi(x,y)$, and $HT_f(x,y,z)$ the estimate after using
  $HT_f(x,y)$ with Theorem \ref{thmintegrate}.
We present our results in Table \ref{tablehtf}.

We recommend this method.

\begin{table*}
  \caption{$HT_f(x,y,z) / \Psi(x,y,z)$, $x=2^{40} $ \label{tablehtf} }
  {
\begin{center}\begin{tabular}{l|lllllll}
\hline
$y$ &
$z=2^{10}$ & $z=2^{12}$ & $z=2^{14}$ & $z=2^{16}$ & $z=2^{18}$ & $z=2^{20}$ & \\
\hline
$2^{2}$ &
1.0969 & 1.0712 & 1.0641 & 1.067 & 1.0728 & 1.0833 & \\
$2^{4}$ &
1.0592 & 1.039 & 1.031 & 1.0353 & 1.0507 & 1.0829 & \\
$2^{6}$ &
1.0414 & 1.0281 & 1.0205 & 1.024 & 1.0566 & 1.1582 & \\
$2^{8}$ &
1.0194 & 1.0165 & 1.0135 & 1.0127 & 1.0296 & 1.1338 & \\
$2^{10}$ &
0.99773 & 1.0069 & 1.007 & 1.0073 & 1.0088 & 1.0413 & \\
$2^{12}$ &
 --- & 1.0152 & 1.0186 & 1.0171 & 1.0141 & 1.0164 & \\
$2^{14}$ &
 --- &  --- & 1.0177 & 1.0186 & 1.0203 & 1.0222 & \\
$2^{16}$ &
 --- &  --- &  --- & 1.021 & 1.0247 & 1.0246 & \\
$2^{18}$ &
 --- &  --- &  --- &  --- & 1.021 & 1.0172 & \\
$2^{20}$ &
 --- &  --- &  --- &  --- &  --- & 1.013 & \\
\hline
\end{tabular}\end{center}
  }
\end{table*}

\begin{table*} 
\caption{$S(x,y,z) / \Psi(x,y,z)$, $x=2^{40} $\label{tablesuz} }
  {
\begin{center}\begin{tabular}{l|lllllll}
\hline
$y$ &
$z=2^{10}$ & $z=2^{12}$ & $z=2^{14}$ & $z=2^{16}$ & $z=2^{18}$ & $z=2^{20}$ & \\
\hline
$2^{2}$ &
0 & 0 & 0 & 0 & 0 & 0 & \\
$2^{4}$ &
0 & 0 & 0 & 0 & 0 & 2.3032 & \\
$2^{6}$ &
0.57462 & 0.60753 & 0.64352 & 0.68488 & 0.7388 & 0.82405 & \\
$2^{8}$ &
0.90295 & 0.90963 & 0.91195 & 0.91382 & 0.92929 & 1.0201 & \\
$2^{10}$ &
0.94404 & 0.93337 & 0.92495 & 0.9175 & 0.9108 & 0.93201 & \\
$2^{12}$ &
 --- & 0.93348 & 0.90717 & 0.89212 & 0.87883 & 0.87083 & \\
$2^{14}$ &
 --- &  --- & 0.90403 & 0.87455 & 0.85967 & 0.84873 & \\
$2^{16}$ &
 --- &  --- &  --- & 0.88064 & 0.85461 & 0.83658 & \\
$2^{18}$ &
 --- &  --- &  --- &  --- & 0.85917 & 0.82827 & \\
$2^{20}$ &
 --- &  --- &  --- &  --- &  --- & 0.83353 & \\

\hline
\end{tabular}\end{center}
  }
\end{table*}

\subsection{Suzuki's Algorithm}

In successive papers, Suzuki \cite{Suzuki2004,Suzuki2006} 
  develops a very fast algorithm, with cost
   $O(\sqrt{\log x \log y})$ operations,
  using the saddle point method to estimate $\Psi(x,y)$.
This is based on good approximations for $\alpha$
  and the prime sums $\zeta,\phi_1,\phi_2$ using the prime number theorem.

For $u>1$, let $\xi$ be the positive solution to the equation
$e^\xi = 1+u\xi$, or equivalently, $\xi = \log(1 + u \xi)$.
This last equation implies that $\xi \approx \log(u\log u)$, and
can be used iteratively to evaluate $\xi$.  
(See \S\ref{secXI} for more information on this point.)


Let $\gamma = 0.57721...$ be Euler's constant.
We now define
\begin{equation}
\alpha_s := 1 - \frac{\xi}{\log y} 
\end{equation}
and
\begin{equation}
  S(x,y):=
    \frac{x^{\alpha_s} e^{\gamma+\int_0^\xi t^{-1}(e^t-1) dt}}{
    \alpha_s \sqrt{2\pi u (1+(\log x)/y)}}.
\end{equation}
Suzuki proves the following \cite[Theorem 1.1]{Suzuki2006}:
\begin{theorem}
Let $\epsilon \le 1/2$.
If $(\log\log x)^{5/3-\epsilon}<\log y<e^{-1}(1-\epsilon)\log x$, then
\begin{equation}
  \Psi(x,y) = S(x,y) (1+o(1)).
\end{equation}
\end{theorem}

Suzuki proposed using the midpoint method to evaluate the 
  integral $\int_0^\xi t^{-1}(e^t-1) dt$.
We used its Maclaurin series $\sum_{n \ge 1} \xi^n / (n \cdot n!)$
\cite[formulas 5.1.10 and 5.1.40]{AS}.
We write $S(x,y,z)$ for the function to estimate
  $\Psi(x,y,z)$ using $S(x,y)$ to estimate $\Psi(x,y)$
  in Theorem \ref{thmintegrate}.

Our results are presented in Table \ref{tablesuz}.
We found the algorithm to be extremely fast, but not accurate for
  small $y$, which is not surprising given the approximations used.

Earlier, in \cite{Suzuki2004}, Suzuki discussed $\hbox{HT}(x,y,\alpha_s)$
  as an approximation to $\Psi(x,y)$.  
Although this is faster than Algorithm HT by a factor of $\log\log x$, 
  it is not as accurate, and so we chose not to test its use.

\subsection{Speed}

Below we give timing results for the algorithms presented above.
This is the total time it took to compute the estimates given
  in the tables in this section.
\begin{center}
\begin{tabular}{lr}
\textit{Algorithm} & \textit{Time in Seconds} \\ \hline
Suzuki & 0.09 \\
Bach and Peralta & 0.40 \\
Ekkelkamp & 0.78 \\
$HT_f$ & 1.90 \\
$HT$ & 20.0 \\
\hline
\end{tabular}
\end{center}
We used the Gnu \texttt{g++} compiler on a unix server
  with an intel CPU.

\section{Numerical Details}

We elaborate on some details from the algorithms presented above.

\subsection{Estimates for $\alpha$}

In this subsection, we give more information about the
function $\alpha(x,y)$, defined implicitly by the equation
$\phi_1(\alpha,y) + \log x = 0$ (here $2\le y \le x$).  
In particular we will prove that
\begin{equation} \label{alphabounds}
 \frac 1 {2 \log x} \le \alpha \le 2.
\end{equation}

We prove the lower bound first.
From the prime number sum for $\phi_1$, we see that for $s > 0$,
$$
-\phi_1 \ge \frac{\log 2}{2^s - 1}.
$$
So $\alpha$, the solution to $\phi_1(\alpha,y) + \log x = 0$,
is lower bounded by the solution $\beta$ to
$$
\frac{\log 2}{2^\beta - 1} = \log x.
$$
Solving, we get
$$
\alpha \ge
\beta = \frac{\log(1 + (\log 2)/(\log x))}{\log 2}
      \ge \frac 1 {2 \log x},
$$
the last inequality holding whenever $x \ge 2$.

Next we show the upper bound.  Let $\zeta$ denote the
Riemann zeta function.
By examining their Dirichlet series, we can see that
$-\phi_1(s,y) \le -\zeta'/\zeta(s)$.  Both sides
are decreasing smooth functions of $s$ on $(0,\infty)$.
It follows that $\alpha$ is upper bounded by the 
solution to $\zeta'(s)/\zeta(s) = \log 2$, which is less
than 2.

More precise information can be found in \cite{HT86}.
In particular,
$$
\alpha = \frac{\log(1 + y/\log x)}{\log y}
         \left\{
         1 + O \left(
               \frac{ \log\log(1+y) }{\log y}
               \right)
         \right\}
$$
holds uniformly for $x \ge y \ge 2$, with the explicit
lower bound
$$
\alpha \ge \frac{\log(1 + y/(5 \log x))}{\log y}.
$$
The strength of this is similar to (\ref{alphabounds}) if
$y$ is fixed and $x \rightarrow \infty$.  

\subsection{Computing $\xi$\label{secXI}}

Here we discuss some numerical methods for solving
$e^\xi = 1+u\xi$, when $u>1$.

Let $f(x) = x - \log(1 + ux)$.  Then $f$ is convex 
on $(0,\infty)$ with a minimum at $x = 1-1/u$.
%
%
%
This gives the lower bound $\xi \ge 1 - 1/u$.
%
%
To get an upper bound, we observe that $e^x > 1 + x + x^2 / 2$,
so $\xi$ is no larger than the positive solution to
$1 + \xi + \xi^2 / 2 = 1 + u\xi$, which is $2(u-1)$.
Using binary search between these bounds, we can get
$\lfloor \xi \rfloor$ plus $d$ bits of its fraction,
with $O(\log u + d)$ evaluations of $f$.

Starting from the defining equation and taking logarithms,
we get
$$
\xi = \log (u \xi + 1).
$$
Suzuki \cite[Lemma 2.2]{Suzuki2004} proves that this iteration, 
starting from $\log u$, is linearly convergent to $\xi$.  
(Here $u > e$ is fixed.)

In practice, we can use the Newton iteration
$$
   \xi := \xi - 
          \frac{(\xi - \log(1 + u\xi))(1 + u\xi)}{1 + u(\xi-1)},
$$
starting with the upper bound $2(u-1)$.  By convexity, the iterates
decrease toward the root.  We tested values of $u$ from 2 to 1000,
and for these $u$, about 5 iterations were enough to
get machine accuracy (about 15D).
(Note that when $x$ is large, $f(x) \sim x$,
so even if we start with a large upper bound, the second iterate
will be much closer to the root.)

Newton iteration does not work well when $u$ is close to 1,
for the following reason.  Suppose that $u = 1 + \epsilon$.  Then, we
have $\epsilon + O(\epsilon^2)  < \xi < 2\epsilon$, making 
$\xi - \log(1 + u\xi)$ vanish to first order in $\epsilon$. 
Thus, when computing this factor, we will lose precision 
due to cancellation.

If special function software is available, $\xi$ can be 
expressed using the Lambert W function.  For example, in the notation
of a well known computer algebra system of Canadian origin,
$\xi = -1/u - \hbox{LambertW}(-1, -e^{-1/u}/u)$.  (The argument $-1$
indicates which branch should be used.)

\subsection{ATM Summation\label{secATM}}

The purpose of the next two subsections is to justify the remark
made earlier that the cost of evaluating the formula HT can be
lowered to $y^{2/3 + o(1)}$.  Here, the ``$o(1)$'' term includes
factors of order $\log x$, so we are implicitly assuming that
$x$ is not outrageously large.

If $f$ is a function defined on the positive integers, it is
\textit{multiplicative} if $f(mn) = f(m)f(n)$.  This is a stronger
requirement than is usual in number theory, where $m$ and $n$
need only be coprime.  The concept of an \textit{additive} function
is defined similarly; we require $f(mn) = f(m) + f(n)$.
We will call $f$ an \textit{ATM function}
(additive times multiplicative) if $f = gh$, where $g$ is
additive and $h$ is multiplicative.

The paper \cite{Bach2007} gave an algorithm that evaluates
the prime sum $\sum_{p \le y} f(p)$, with $f$ an ATM function,
in $y^{2/3+o(1)}$ steps.  
This generalized a previously known
result, also explained in \cite{Bach2007}, in which $f$ could
be multiplicative.  

We now explain how these summation algorithms can be used
to evaluate $\phi_1$ and $\log \zeta$.   The basic idea is
to approximate each of these by a ``small'' number of ATM or 
multiplicative prime sums.  With $\log \zeta$ in hand,
we can exponentiate to get $\zeta$. 

We first assume $s > 0$. 

Let us consider $\phi_1$ first.
By summing geometric series, we see that
$$
- \phi_1 = \sum_{p \le y} \frac{\log p}{p^s - 1}
         = \sum_{k \ge 1} \sum_{p \le y} \frac{\log p}{p^{ks}}.
$$
Note that each inner sum involves an ATM function.  We will
restrict the outer sum to $1 \le k \le N$, and choose $N$ to 
make the truncation error,
$$
\sum_{k \ge N+1}
\sum_{p \le y} \frac{\log p}{p^{ks}}.
$$
small.

If we interchange the order of summation, allow all $p \ge 2$,
and sum geometric series, we can express the truncation error
as
$$
\sum_{p \ge 2} \frac{\log p}{p^{Ns}(p^s - 1)}.
$$
Using the globally convergent Maclaurin series for $p^s$, we
see that $1/(p^s - 1) \le 1/(s \log p)$.  If we plug this in, 
the $\log p$ factors cancel and we get the upper bound
$$
\frac 1 s \sum_{p \ge 2} \frac 1 {p^{Ns}}
\le
\frac 1 s \left( \frac 1 {2^{Ns}} +
\int_{2}^\infty \frac {dt} {t^{Ns}} \right)
\le
\frac 3 {s 2^{Ns}},
$$
provided that $Ns \ge 2$.
For us, $s \ge 1/(2 \log x)$, and with this additional 
assumption we get
$$
\hbox{[truncation error]} \le \frac {6 \log x} {2^{Ns}}.
$$

Thus, to achieve truncation error less than $2^{-d}$, we can use
$N = \Theta((\log x)(\log\log x + d))$.  

Similarly, we can use
$$
\log \zeta(s,y) = - \sum_{p \le y} \log(1 - p^{-s})
= \sum_{k \ge 1} \frac 1 k \sum_{p \le y} \frac 1 {p^{ks}}.
$$
Now each inner sum involves a multiplicative function.
If we use only the inner sums with $k \le N$, a similar
analysis shows that choosing $N = \Theta((\log x)(\log\log x + d))$
will keep the truncation error below $2^{-d}$.

\subsection{Numerical Differentiation}

\def \ld {\hbox{ld}}

In this subsection, we explain how to evaluate
$$
\phi_2(s,y) = 
\sum_{p \le y} \frac {\log^2 p \cdot p^s}{(p^s - 1)^2},
$$
for use in Algorithm HT.  There is no obvious way
to reduce this to the kind of sums treated in \cite{Bach2007},
so we will approximate it by a difference.

%
%
%

Using balanced numerical differentiation \cite[p. 297]{CdB},
we have
$$
\phi_2(s,y) = 
\frac{\phi_1(s+h,y) - \phi_1(s-h,y)}{2h} + \epsilon,
\qquad
\epsilon = \frac{h^2}{6} \phi_4(\eta,y)
$$
for some $\eta \in [s-h, s+h]$.  Let us determine
how much precision will be necessary to deliver $d$ bits of
$\phi_2(s,y)$ accurately, when we use this formula.

After differentiating the sum for $\phi_2$ twice, we see that
$$
\phi_4(s,y) 
= \sum_{p \le y} \frac {p^s \log^2 p}{(p^s - 1)^2} 
\times
\frac {(p^{2s} + 4p^s + 1) \log^2 p}{(p^s - 1)^2} .
$$
Observing that $(t^2 + 4t + 1)/(t-1)^2$ is decreasing for $t > 1$, 
we get the estimate
$$
\frac{h^2}{6} \phi_4(s,y)
\le \frac{h^2 \log^2 y}{6}
\left(\frac {2^{2s} + 4\cdot 2^s + 1}{(2^s - 1)^2}\right) \phi_2(s,y).
$$
If $0 \le s \le 2$, the factor in parentheses is bounded by
$15/s^2$.  (Numerically, anyway.) So, if we could use exact 
arithmetic, numerical differentiation would give us
$$
\hbox{[relative error]} \le \frac 5 2 \ \frac{h^2 \log^2 y}{s^2}
                        \le 10 h^2 \log^2 x \log^2 y.
$$

However, we don't have exact arithmetic, so we must also analyze the
loss of precision due to cancellation.  For this, we use an ad hoc
theory.  When $h$ is small, the number of bits lost, when using the 
balanced difference formula to compute $\phi_1'(s)$, is about
$$
-\log_2 \left| \frac{\phi_1(s+h) - \phi_1(s-h)} {\phi_1(s)} \right|,
$$
since dividing by $h$ causes no loss of precision.  
Note that this is a centered version of the usual relative error formula.
Since $\phi_1(s+h) - \phi_1(s-h) \sim 2h \phi_1'(s)$, we must bound 
the logarithmic derivative of $\phi_1$, or, what is the same thing, 
relate $\phi_2 = \phi_1'$ to $\phi_1$.

In our case, we require a lower bound 
for $\phi_2$.  
Then since $t/(t-1)$ is decreasing
on $(1,\infty)$,
$$
\phi_2(s,y) 
= \sum_{p \le y} \frac{\log p}{p^s - 1} 
\times
\frac {p^s \log p}{p^s - 1}
$$
$$
\ge
\sum_{p \le y} \frac{\log p} {(p^s - 1)} \log p 
\ge
-\phi_1(s,y)\log 2.
$$
%
%
%
Therefore,
$$
\phi_1(s+h,y) - \phi_1(s-h,y)
\sim
2h \phi_2(s,y) \ge - \phi_1(s,y) \cdot 2 h \log 2,
$$
as $h \rightarrow 0$,

Let us now translate these results into practical advice.
Suppose our goal is to obtain $\phi_2$ to $d$ bits of precision,
in the sense of relative error.  By the exact arithmetic formula,
we should choose $h \le 2^{-d/2} / (\sqrt{10} \log^2 x)$.
Then we need to use $\log_2 h^{-1} + O(1)$ guard bits in 
our computation.  Put more crudely, unless $x$ is very large,
doubling the working precision should be enough, if 
we select $h$ properly.

The following example indicates that the theory above
is roughly correct.  Suppose we want 10 digits of 
$\phi_2$, for $x = 10^6$, $y = 10^3$, and
$s = 1/(2 \log x) = 0.03619...$\ .  Our recipe
allows us to take $h = 10^{-7}$.  With
17 digit arithmetic, we obtained
\begin{eqnarray*}
\mbox{numerical derivative} & =&   127790.77386350000 \\
\mbox{summation for $\phi_2$} & =& 127790.77386041727 
\end{eqnarray*}
which agree to 11 figures.

\section{Conclusion and Future Work}

In summary, we recommend the
  estimate $HT_f(x,y,z)$ of \S\ref{htfast} for approximating $\Psi(x,y,z)$.
We feel it gives high accuracy while retaining sufficient speed to
  be very practical.

For future work, we hope to generalize our results
  to $2$ or more large primes.
We also hope to further examine estimates of the form of (\ref{futurework}).

\section{Acknowledgments}
We want to thank the referees, whose comments helped improve this paper.

Abstract presented at the AMS-MAA Joint Mathematics Meetings,
January 2012, Boston MA,
and at the CMS Summer Meeting, June 2013, Halifax Nova Scotia.  

Supported in part by grants from the 
Holcomb Awards Committee, NSF (CCF-635355), and ARO (W911NF-09-1-0439).

\bibliographystyle{abbrv}

%
\end{document}